\documentclass[11pt]{article}

\usepackage{amsthm, amsbsy, amssymb, amsfonts}
\usepackage{fullpage,times}

\newtheorem{theorem}{Theorem}
\newtheorem{lemma}[theorem]{Lemma}

\newtheorem{definition}[theorem]{Definition}
\newtheorem{proposition}[theorem]{Proposition}

\def\E{\mathop{\rm E}}
\newcommand{\defeq}{\stackrel{\Delta}{=}}
\newcommand{\cN}{{\mathcal N}}
\newcommand{\dist}{\mathsf{dist}}

\newcommand{\cT}{\mathcal{T}}

\newcommand{\cA}{\mathcal{A}}

\newcommand{\depth}{\mathsf{depth}}
\newcommand{\parity}{\mathsf{Parity}}
\newcommand{\majority}{\mathsf{Parity}}
\newcommand{\val}{\mathsf{val}}

\begin{document}

\title{A Sampling Technique of Proving Lower Bounds for Noisy Computations\footnote{A preliminary version of this work appeared 
in the Proceedings of the
 49th Annual IEEE Symposium on Foundations of Computer Science, 2008, pp. 394-402.}}

\author{Chinmoy Dutta\thanks{Twitter Inc., San Francisco, USA. 
email: {\tt chinmoy@twitter.com}. The work was done while this author was at
Tata Institute of Fundamental Research, Mumbai, INDIA.} \\
\and
Jaikumar Radhakrishnan\thanks{Tata Institute of Fundamental Research, Mumbai, INDIA. 
email: {\tt jaikumar@tifr.res.in}}}

\date{}
\maketitle

\begin{abstract}
We present a technique of proving lower bounds for noisy computations. 
This is achieved by a theorem connecting computations on a kind of 
randomized decision trees and sampling based algorithms. 
This approach is surprisingly powerful, and applicable to several 
models of computation previously studied. 

As a first illustration we show how all the results of Evans and Pippenger 
(SIAM J. Computing, 1999) for noisy decision trees, some of which were derived 
using Fourier analysis, follow immediately if we consider the sampling-based
algorithms that naturally arise from these decision trees.

Next, we show a tight lower bound of $\Omega(N \log\log N)$ on the number of
transmissions required to compute several functions (including the
parity function and the majority function) in a network of $N$
randomly placed sensors, communicating using local transmissions, and
operating with power near the connectivity threshold. This result
considerably simplifies and strengthens an earlier result of Dutta, Kanoria
Manjunath and Radhakrishnan (SODA 08) that such networks cannot compute 
the parity function reliably with significantly fewer than $N\log \log N$ transmissions.
The lower bound for parity shown earlier made use of special properties of the parity function and is
inapplicable, e.g., to the majority function. In this paper, we use our approach to
develop an interesting connection between computation of boolean
functions on noisy networks that make few transmissionss, and algorithms
that work by sampling only a part of the input. It is straightforward
to verify that such sampling-based algorithms cannot compute the
majority function. 
\end{abstract}

\section{Introduction}
\label{sec:introduction}

We present a novel technique for analyzing randomized decision trees. This
method does not depend upon the specfic function being computed by the decision tree
and can be applied for proving lower bounds in various models for a variety of functions. 

We introduce the technique in the simplistic setting of $\epsilon$-noisy decision trees. 
$\epsilon$-noisy decision trees can be viewed as a simple kind of randomized decision trees. 
As an application of our technique, we show how it provides elementary and unified proofs 
of all the lower bounds of Evans and Pippenger~\cite{Evans99} for 
average noisy decision tree complexity of several types of functions. 
Their work introduced the notion of {\em noisy leaf complexity}, which was analyzed using Fourier methods.
However, as we show, our technique yields elementary arguments that places these
results in a compact and unified framework.

We then use our technique to derive the main result of this paper - 
a lower bound for wireless sensor networks.
This simplifies and extends a lower bound of\cite{Dutta08, Dutta08-journal} 
for this model. A wireless sensor network consists of sensors that collect and
cooperatively process data in order to compute some global
function. The sensors interact with each other by transmitting
wireless messages based on some protocol. The protocol is required to
tolerate errors in transmissions since wireless messages typically are
noisy.

In the problem we study, there are $n$ sensors, each with a boolean
input, they are are required to cooperatively compute some function of
their inputs. The difficulty of this task, of course, depends on the
noise and the connectivity of the network. In this paper, we assume
that each bit sent is flipped (independently for each receiver) with
probability $\epsilon>0$ during transmission. As for connectivity, we
adopt the widely used model of random planar networks. Here the
sensors are assumed to be randomly placed in a unit square. Then each
transmission is assumed to be received (with noise) by the sensors
that are within some prescribed radius of the sender. The radius is
determined by the amount of power used by the sensors, and naturally
one wishes to keep the power used as low as possible, perhaps just
enough to ensure that the entire network is connected. It has been
shown by Gupta and Kumar~\cite{Gupta00} that the threshold of
connectivity is $\theta\left(\sqrt{\frac{\ln n}{n}}\right)$ (with a
radius much smaller than this the network will not be connected almost
surely, and with radius much larger it will be connected almost
surely).

It was shown by Dutta, Kanoria, Manjunath and Radhakrishnan~\cite{Dutta08, Dutta08-journal},
that computing the parity of the inputs requires $\Omega(n \log \log
n)$ transmissions. This result showed that the protocol presented by
Ying, Srikant and Dullerud~\cite{Ying06} for computing the sum of all
the bits (and hence any symmetric functions of these bits) is
optimal. The lower bound argument in \cite{Dutta08, Dutta08-journal} depended strongly on the
properties of the parity function. In particular, the argument was not
applicable for showing superlinear lower bounds for the majority and
other symmetric functions in this model. Using the technique presented in this paper, we 
can now give optimal lower bounds for several functions, including the majority function.

The starting point of our method is the connection between noisy
wireless network protocols, and certain randomized decision trees.
This connection was derived in \cite{Dutta08, Dutta08-journal}, and it in turn made crucial
use of a result in \cite{Goyal05}. The lower bound for parity in \cite{Dutta08, Dutta08-journal}
was derived by rearranging the randomized decision trees obtained from
wireless protocols computing parity. In this work, we show that computations on decision
trees that arise in our context can be simulated by randomized
algorithms that leave a non-trivial fraction of their inputs unread.
Once this is established, it is relatively straightforward to
conclude that several functions cannot be computed in this model.

In order to state our result formally we need a formal definition of
the model of noisy communication networks, which we now reproduce
from~\cite{Dutta08}.
 
\begin{definition}[Noisy communication network, protocol]
\label{def:noisy-network-protocol} 

A communication network is an undirected graph $G$ whose vertices
correspond to processors and edges correspond to communication links.
A message sent by a processor is received by all its neighbors.
\begin{description}
\item[Noise:] In an $\epsilon$-noise network, the messages are
subjected to noise as follows. Suppose processor $v$ sends bit $b$ in
time step $t$. Each neighbor of $v$ then receives an independent
noisy version of $b$; that is, the neighbor $w$ of $v$ receives the
bit $b \oplus \eta_{w,t}$, where $\eta_{w,t}$ is an $\epsilon$-noisy
bit (that takes the value $1$ with probability $\epsilon$ and $0$
with probability $1-\epsilon$), these noisy bits being mutually
independent.

\item[Input:] An input to the network is an assignment of bits to the processors, and is formally
an element of $\{0, 1\}^{V(G)}$.
 
\item[Protocol:] A protocol on $G$ for computing a function
$f:\{0,1\}^{V(G)} \rightarrow \{0,1\}$ works as follows. The
processors take turns to send single bit messages, which are received
only by the neighbors of the sender.  In the end, a designated
processor $v^* \in V(G)$ declares the answer. The cost of the protocol
is the total number of bits transmitted.  A message sent by a
processor is a function of the bits that it possesses until then.  The
protocol with cost $T$ is thus specified by a sequence of vertices
$\langle v_1,v_2,\ldots,v_T\rangle$ and a sequence of $T$ functions
$\langle g_1,g_2,\ldots,g_T \rangle$, where $g_t:\{0,1\}^{j_t}
\rightarrow \{0,1\}$ and $j_t$ is the number of bits received by $v_t$
before time step $t$ (plus one if $v_t$ is an input
processor). Furthermore, $v_T=v^*$, and the final answer is obtained
by computing $g_T$. Note that in our model the number of transmissions
is the same for all inputs.

\item[Error:] Such a protocol is said to be a $\delta$-error protocol,
if for all inputs $x \in \{0,1\}^{V(G)}$, $\Pr[\mathsf{output} = f(x)] \geq 1-\delta$.
\end{description}
\end{definition}

In this paper, we consider networks that arise out of random placement 
of processors in the unit square.

\begin{definition}[Random planar network]
A random planar network $\cN(N,R)$ is a random variable whose values
are undirected graphs. The distribution of the random variable depends
on two parameters: $N$, the number of vertices, and $R$, the
transmission radius.  The vertex set of $\cN(N,R)$ is $\{P_1, P_2,
\ldots , P_N\}$. The edges are determined as follows. First, these
processors are independently placed at random, uniformly in the unit
square $[0,1]^2$. Then,
\[ E(\cN) = \{ (P_i,P_j):  \dist(P_i,P_j) \leq R\}.\]
\end{definition}

Our main result is the following.
\begin{theorem}[Lower bound for majority]
\label{thm:lb-majority}
Let $R \leq N^{-\frac{1}{3}}$.  Let $\delta, \epsilon \in
(0,\frac{1}{2})$.  Then, with probability $1-o(1)$ (over the placement
of processors) every $\delta$-error protocol on $\cN(N,R)$ with
$\epsilon$-noise for computing the majority on $N$ bits requires
$\Omega(N \log \log N)$ transmissions.
\end{theorem}

\paragraph{Remarks:} 
\begin{itemize}
\item It was conjectured in \cite{Dutta08} that one cannot approximate
the sum to within an additive error of $N^\alpha$ (for some $\alpha >
0$) using $O(N)$ transmissions. Using the techniques of this paper,
we can prove this conjecture (details omitted from this paper).

\item Typical protocols in wireless networks operate by computing
using broadcasts in cells, where there is complete connectivity
between the processors, and then aggregating the information across
the cells (e.g. ~\cite{Kanoria07}, ~\cite{Ying06}). Here it
makes sense to consider functions of the form
$f(g(X_1),g(X_2),\ldots,g(X_k))$, where $X_i$ is the part of the input
that falls in cell $i$, $f$ is some symmetric boolean function and $g$
is some function to be computed inside each cell. For example, if $f$
and $g$ are both parity, then this corresponds to the parity function
on the entire input. Our techniques show that for most symmetric
functions, and all $g$ with high sensitivity, one requires $\Omega(N
\log \log N)$ transmissions. (Details omitted.)

%
\end{itemize}

As stated earlier, our technique yields simple proofs of previous lower
bound results on noisy decision trees.

\begin{definition}[Noisy decision trees]
A {\em boolean decision tree $\cT$} for input $x = \langle<x_1, \ldots
,x_n\rangle \in \{0,1\}^n$ is a binary tree in which each internal
node $v$ has a label $l(v) \in [1,n]$ and each leaf $\ell$ has a value
$val(\ell) \in \{0,1\}$. The two outgoing edges of each internal node
are labelled by the value $0$ and $1$.  The computation of $\cT$ on
input $x$ is the unique path starting at the root of the tree and
continuing up to a leaf as follows: at internal node $v$, the outgoing
edge labelled with $x_{l(v)}$ is chosen to get to the next node.  The
result of the computation is $\val(\ell)$ where $\ell$ is the leaf
reached by the computation. In an {\bf $\epsilon$-noisy boolean
decision tree}, at each internal node, the incorrect outgoing edge is
chosen with probability $\epsilon$ independent of former
choices. Equivalently, each internal node $v$ is assigned a binary
random variable $\eta_v$ that takes the value $1$ with probability
$\epsilon$ independently. Then, on reaching internal node $v$, the
outgoing edge labelled $x_{l(v)} \oplus \eta_v$ is used to determine
the next node.  For each input $x$, the computation path and the value
$\cT(x)=\ell(x)$ output by the tree is a random variable.  Let
$\depth_{\epsilon,\delta}(f)$ be the minimum depth of an
$\epsilon$-noisy decision tree $\cT$ such that
$\Pr[\cT(x) \neq f(x)] \leq \delta$, where the probability is
over the input $x$ chosen uniformly from $\{0,1\}^n$ and the 
internal randomness of $\cT$.
\end{definition}

The main results of Evans and Pippenger~\cite{Evans99} were 
(a) $\depth_{\epsilon,\frac{3}{4}}(\parity_n) 
       = \Omega(n (\log n)/\log (1/\epsilon))$,
(b) $\depth_{\epsilon,\frac{3}{4}}(f)  
            = \Omega(n (\log n)/log (1/\epsilon))$ and
      for almost all functions;
(c) $\depth_{\epsilon,\frac{3}{4}}(f) =
      \Omega(n \log s)$
      if $f$ is $n\left(1-\frac{1}{s}\right)$-resilient.

In Section~\ref{sec:evans-pippenger}, we provide a simple proof
of the lower bounds.

\subsection{Related work}

Noisy broadcast models have been studied in the past where all sensors
receive all messages (with independent noise).
Gallager~\cite{Gallager88} showed a remarkable protocol to collect all
input bits at one sensor using $O(N\log \log N)$
transmissions. Clearly, this give the same upper bound for computing
any function of the input bits. Several other works have focussed on
constructing protocols for specific functions in variants of the noisy
broadcast model, e.g., Feige and Raghavan~\cite{Feige00},
Newman~\cite{Newman04}, Kushilevitz and Mansour~\cite{Kushilevitz98},
and Goyal, Kindler and Saks~\cite{Goyal05}.  Using an insightful
combination of information-theoretic and fourier based methods, Goyal,
Kindler and Saks~\cite{Goyal05} showed that Gallager's protocol was
the best possible for collecting all the bits. 

In sensor networks, considerations of power impose stringent limits on
the transmission radius. In this paper, we study networks arising from
random placement of sensors with transmission radius around the
threshold required to ensure connectivity.  As mentioned above, in
this model Ying, Srikant and Dullerud~\cite{Ying06} devised a
protocol for computing the sum using $O(N\log \log N)$ transmissions.
Kanoria and Manjunath~\cite{Kanoria07} showed a protocol with $O(N)$
transmissions to compute the OR function. Making crucial use of a
result of Goyal, Kindler and Saks~\cite{Goyal05}, it was shown in~\cite{Dutta08, Dutta08-journal} 
that computing parity requires $\Omega(N\log \log N)$ transmissions
which was then extended to several other functions in the initial presentation of 
this work\cite{Dutta08b}. 

Unlike in the model of noisy wireless networks, several lower bounds
results have appeared in the literature on noisy decision trees.
Reischuk and Schmeltz~\cite{Reischuk91} showed that almost all boolean functions
of $N$ arguments require $\Omega(N \log N)$ queries. Feige, Peleg,
Raghavan, and Upfal~\cite{Feige94} showed an $\Omega(N \log N)$ lower bound
for the parity function. Evans and Pippenger~\cite{Evans99} presented
arguments to show that these results also hold in the average case.

\subsection{Techniques}


We now present an overview of the proof technique used to derive the
lower bounds in this paper. A very detailed discussion 
of all the techniques can be found in the Phd thesis~\cite{Dutta09}.

The proof of our main result,like the proof in \cite{Dutta08}, first converts computations
on noisy broadcast networks to computations on randomized decision
trees. 

\begin{definition}[Randomized decision tree]
\label{def:randomizedtree}
A randomized decision tree is a model for processing inputs in $\{0,
1\}^n$. For an internal node $v$ of the tree let $v_L$ be its left
child and $v_R$ its right child.  Each internal node $v$ of the tree
is labelled by a pair $\langle{i_v,g_v}\rangle$, where $i_v \in
\{1,2,\ldots,n\}$, and $g_v:\{0,1\} \rightarrow \{v_L,v_R\}$ is a noisy
function, whose output depends on its input and some internal
randomness that is independent for different noisy computations
performed in the tree. Once an input $x=\langle x_1,x_2, \ldots,
x_n\rangle \in \{0, 1\}^n$ is fixed, the (random) output of the
tree is determined by the following natural computation. We start at
the root, and when we arrive at an internal node $v$, we
determine the next vertex by evaluating $g_v(x_{i_v})$. The (random)
output of the tree on input $x \in \{0, 1\}^n$ is the 0-1 label
of leaf reached.
\end{definition}
\newcommand{\overlap}{\mathsf{overlap}}
\newcommand{\uncertainty}{\beta}

Our lower bound for noisy broadcast networks will follow from a lower
bound we show for randomized decision trees that arise from them. A
central notion in our analysis, is the amount of uncertainty about
each variable that remains at the end of the computation. We will use
the notion of overlap between distributions to quantify this
uncertainty.

\begin{definition}[Overlap between distributions] 
\label{def:overlap}
Let $D_0$ and $D_1$ be distributions on some set $L$. The overlap between $D_0$ and $D_1$
is given by $\overlap(D_0,D_1) = \sum_{\ell \in L} \min\{D_0(\ell),
D_1(\ell)\}.$ Note that this quantity is directly related to the
$\ell_1$ distance between $D_0$ and $D_1$: $\|D_0,D_1\|_1 =
2(1-\overlap(D_0,D_1))$.  
\end{definition}

\begin{definition}[Uncertainty] 
\label{def:uncertainty}
Let $\cT$ be a randomized
decision tree for single bit inputs, where each internal node computes
a function based $x\in\{0,1\}$. Let $D_0$ be the distribution on
$\cT$'s leaves when $x=0$, and let $D_1$ be the distribution when
$x=1$. Then, the the uncertainty of $\cT$ about $x$, denoted by
$\uncertainty(\cT)$, is given my $\overlap(D_0,D_1)$. We will
generalize this notion to trees with inputs in $\{0,1\}^k$.  The
uncertainty about $x_i$ is given by
$\uncertainty_i(\cT)\defeq\min_{\cT'} {\overlap(\cT')}$,
where the minimum is taken over all trees $\cT'$ 
obtained from $\cT$ by 
\begin{enumerate} 
\item[(a)] retaining the functions at nodes that query $x_i$ and 
\item[(b)] replacing the functions at nodes that query variables
    $x_j$ ($j\neq i$) by arbitrary constant functions;
\end{enumerate}
Finally, define $\uncertainty(\cT) = \min_{i} \uncertainty_i(\cT)$.
\end{definition}

For trees with multiple inputs, $\uncertainty_i(\cT)$ denotes the
uncertainty that remains about $x_i$ no matter how the decisions are
made at nodes that query other variables.  We can now state the
following crucial connection between broadcast protocols and
randomized decision trees which follows from arguments in \cite{Dutta08}.

\begin{proposition} 
\label{prop:protocol-to-tree}
If there is a broadcast protocol with $N d$ broadcasts for computing a
the majority of $N$ bits, then there is a randomized decision tree
$\cT$ for computing the majority of $N'\geq \sqrt{N}$ bits
such that $\beta(\cT) \geq \exp(-\exp(O(d)))$. In particular,
if $d \leq \frac{1}{C}\log \log N$ for some constant $C$, 
then $\beta(\cT) \geq N^{-\frac{1}{10}}$. 
\end{proposition}

The main contribution of this work is a technique for analyzing
randomized decision trees with non-trivial uncertainty. Roughly, we
show that if the uncertainty about a variable is non-trivial, then the
computation can be performed by leaving the variable unread with some
non-trivial probability. For example, consider a trivial tree with one
root and two leaves. At the root we read a variable $x\in\{0,1\}$, and
moves to the left child with probability $\frac{1}{2}+
(-1)^x\epsilon$. This tree can be simulated as follows. With
probability $1-2\epsilon$, we do not read $x$ at all, and move left or
right with equal probability. With probability $2\epsilon$ we read $x$
and move left if $x=0$ and move right otherwise. If all nodes are of
this kind and the tree has small depth, then it is not hard to see
that we can simulate its computation leaving several variables unread.
This is precisely the situation in the model of decision trees studied
by Evans and Pippenger~\cite{Evans99}, which explains why our technique is
effective there. However, applying this idea to the randomized
decision trees guaranteed by Proposition~\ref{prop:protocol-to-tree}
requires more careful analysis.  The detailed argument is presented in
Section~\ref{sec:main-theorem}.

\section{The Evans-Pippenger lower bounds revisited}
\label{sec:evans-pippenger}

In this section, we show that the three average case lower bounds of
Evans and Pippenger~\cite{Evans99} follow immediately by considering
sampling-based algorithms that arise naturally from noisy decision
trees. 

\begin{definition}[$(r,\delta)$-sampling algorithm] 
We say that a randomized algorithm $\cA$ is an $(r,\delta)$-sampling
algorithm if 
$\displaystyle
\Pr[\mbox{$\cA$ leaves at least $r$ variables unread}] \geq 1- \delta.$
\end{definition}

\begin{definition}[Robust function]
We say that $f:\{0,1\}^n \rightarrow \{0,1\}$ is an 
$(r, \gamma)$-robust function, if for every subcube $L$ of dimension $r$,
\[ \frac{1}{2} - \gamma \leq \Pr_{x \in L}[ f(x) = 1 ]
                        \leq \frac{1}{2} + \gamma.\]
\end{definition}

Once these definitions are in place, the proofs the results of
Evans and Pippenger~\cite{Evans99} follow easily from the following
observations. 

\begin{enumerate}
\item
$\epsilon$-noisy boolean trees of small depth 
can be simulated by randomized algorithms that typically 
leave many variables unread.

\begin{lemma}
\label{lm:noisy-to-sampling}
Suppose $\cT$ is an $\epsilon$-noisy boolean decision tree with $n$
variables and depth at most $kn$.  Then there is a randomized
algorithm $\cA$ that on all inputs simulates the computation on $\cT$ (producing the
same distribution on the leaves), and with probability
at least $1 - \exp(-\frac{\epsilon^{2k}n}{16})$, leaves at least
$\frac{\epsilon^{2k}}{4}n$ variables unread.
\end{lemma}

This lemma is the key to the analysis in this section. We present its
elementary proof below.

\item Our next observation states that randomized algorithms of the
kind promised by the above theorem cannot compute a robust function
with small error.

\begin{lemma} \label{lm:add-error}
If $\cA$ is an $(r,\delta)$-sampling algorithm computing 
an $(r,\gamma)$-robust function $f$. Then,
\[ \Pr_{x \in \{0,1\}^n}[\cA(x) = f(x)] 
     \leq \frac{1}{2} + \delta + \gamma.
\]
\end{lemma}

\item The three types of functions considered by Evans and
Pippenger are robust.

\begin{lemma}\label{lm:robust-functions}
\begin{enumerate}
\item[(a)] Almost all functions $f:\{0,1\}^n \rightarrow \{0,1\}$ are $(6\log\log n, o(1))$-robust.
\item[(b)] The parity function is $(1,0)$-robust.
\item[(c)] A $t$-resilient function is $(n-t,0)$-robust. 

\end{enumerate}
\end{lemma}

\end{enumerate}

From Lemma~\ref{lm:noisy-to-sampling} it follows immediately that any
$\epsilon$-noisy decision tree of depth at most $\frac{n}{2}
\frac{\log(n/(8r))}{\log(1/\epsilon)}$, can be simulated using an
$\left(r, \frac{1}{8}\right)$-sampling algorithm. Then, from
Lemma~\ref{lm:add-error} it follows that any such tree for an
$(r,\delta)$-robust function makes error at least
$\frac{1}{2}+\frac{1}{8} + \delta$.  By combining this with
Lemma~\ref{lm:robust-functions}, we obtain the following.

\begin{theorem}[Evans and Pippenger~\cite{Evans99}]
\begin{enumerate}
\item $\depth_{\epsilon,\frac{3}{4}}(f) = \Omega(n (\log n)/\log (1/\epsilon))$; for almost all functions;

\item $\depth_{\epsilon,\frac{3}{4}}(\parity_n) = \Omega(n (\log n)/\log (1/\epsilon))$;

\item $\depth_{\epsilon,\frac{3}{4}}(f) = \Omega(n \log s)$ if $f$ is $n\left(1-\frac{1}{s}\right)$-resilient.
\end{enumerate}
\end{theorem}

It remains to prove the lemmas claimed above.

\begin{proof}[Proof of Lemma~\ref{lm:noisy-to-sampling}]
The randomized algorithm works by simulating the computation by the
noisy boolean decision tree starting at the root. The algorithm has a
boolean random variable $B_v$ for each internal node $v$ of the
decision tree. Each $B_v$ takes the value $1$ with probability
$2\epsilon$ independently. At internal node $v$ of the tree, if $B_v =
1$, the sampling algorithm chooses one of the outgoing edges with
probability half each, to get to the next node. If $B_v = 0$, the
algorithm reads the value of the input variable $x_{l(v)}$ (without
any error), and chooses the outgoing edge with that label. On reaching
leaf $\ell(x)$, which is a random variable, the algorithm outputs
$\val(\ell(x))$ as the result. It is easy to see that for any input
$x$, the distribution on leaves reached by the sampling algorithm is
exactly the same as that reached by the noisy decision tree.
 
Fix an input $x$ and a leaf $\ell$ reached by the simulation. We will
show that conditioned on arriving at this leaf, the algorithm leaves
at least $\frac{\epsilon^{2k}}{4} n$ variables unread with high
probability. Suppose the variable $x_i$ appears $k_i$ times in $\cT$
on the path from the root to $\ell$. Then, the probability that $x_i$
is not read conditioned on the computation reaching $\ell$ is at least
$\left(\frac{\epsilon}{1-\epsilon}\right)^{k_i} \geq \epsilon^k_i$.
Since the depth of the tree is at most $nk$, there are at least
$\frac{n}{2}$ variables that appear at most $2k$ times on the path to
$\ell$. Each of these variables is independently left unread with
probability at least $\epsilon^{2k}$.  Using the Chernoff bound, with
probability at least $1 - \exp(-\frac{\epsilon^{2k}n}{16})$, the
algorithm leaves at least $\frac{\epsilon^{2k}n}{4}$ variables unread.
Since this claim is true conditioned on each leaf, it also holds
overall.
\end{proof}

\begin{proof}[Proof of Lemma~\ref{lm:add-error}]
The probability that fewer than $r$ variables are left unread is at
most $\delta$. Conditioned on the algorithm leaving $r$ variables
unread, the probability that its output is correct is at most
$\frac{1}{2}+\gamma$ because $f$ is $(r,\gamma)$-robust.
\end{proof}

\begin{proof}[Proof of Lemma~\ref{lm:robust-functions}]
The second and third claims follow immediately from definitions.
We justify the first claim using the following routine calculation.
Consider a $d$-dimensional subcube of the boolean hypercube
$\{0,1\}^n$. Pick a random function $f$, and let $X$ be the random
variable denoting the number of points in the subcube where $f$ takes
the value $1$. We have $\E[X] = 2^d/2$. Let $t =
\frac{1}{2^{d/3}}$. By the Chernoff bound, $\Pr[|X-\E[X]| > t 2^d] < 2
\exp(- t^2 2^d) = 2 exp(- 2^{d/3})$. Taking the union bound over all
subcubes of dimension $d$, the probability that $f$ has a bias of more
than $t = \frac{1}{2^{d/3}}$ on any such subcube is at most ${n
\choose d} \left( 2 \exp(- 2^{d/3}) \right)$.  Thus, for $d \geq 6
\log \log n$, with probability $1 - o(1)$, $f$ has $o(1)$ bias on
every subcube of dimension $d$.
\end{proof}
\section{Proof of main Theorem}
\label{sec:main-theorem}

In this section we prove our lower bound on the number of transmissions
needed to compute majority in a noisy wireless network. 
By Proposition~\ref{prop:protocol-to-tree}, it is enough
to show that randomized decision trees with high uncertainty cannot
compute majority.  We will first show how such trees can be simulated
by sampling algorithms.  The result follows from this
because it is straightforward to verify that sampling algorithms
that leave a super-constant number of variables unread cannot
compute majority with low error. 

\begin{definition}[Sampling-based algorithm]
By a sampling-based algorithm for computing a function $f:\{0,1\}^n
\rightarrow\{0,1\}$, we mean an algorithm of the following kind.  In
the sampling phase, the algorithm uses $n$ {\em sampling
probabilities} $q_1,q_2,\ldots,q_n \in \{0,1\}$. Given an input
$x=\langle{x_1,x_2,\ldots,x_n}\rangle$, the algorithm constructs a
string $y \in \{0,1, \star\}^n$ from $x$, by independently replacing
$x_i$ by a $\star$ with probability $1-q_i$. In the second phase, it
declares its guess for $f(x)$ based on $y$ alone.
\end{definition}

The main part of the argument is contained in the following theorem.

\begin{theorem} 
\label{thm:randomized-to-sampling}
Let $\cT$ be a randomized decision tree with inputs
from $\{0,1\}^k$ that computes a function $f$ with error at most
$\delta$.  Then, there is a sampling algorithms $\cA$ which
independently samples the $i$-th variable with probability
$q_i=1-\uncertainty_i(\cT)$, and computes $f$ with error at most
$\delta$.  
\end{theorem}

\paragraph{Preliminaries.} 
First, we need some notation.  For an internal node $v$ of the tree,
let $v_L$ denote its left child, and $v_R$ its right child.  Suppose
the input $x_i$ is queried at a $v$. The tree $\cT$ specifies the
probabilities for the computation to move to each child for each
possible value of $x_i$. 
For the node $v$, let $\chi_v$ be
the event that the computation on $\cT$ reaches the node $v$. Clearly,
$\chi_v$ is the intersection of independent events $\{\chi_{v,i}:
i=1,2,\ldots,k\}$, where $\chi_{v,i}$ is the event that the
computation reaches node $v$ assuming and the choices at all nodes not
labelled $i$ don't leave the path. Note that the probability of
$\chi_{v,i}$ depends only on the value of $x_i$; let $p_{v,i}(z)$ be
this probability. Then, the probability of the computation reaching
the node $v$ on input $x=\langle x_1,\ldots,x_k \rangle$ is precisely
$\prod_{i=1}^k p_{v,i}(x_i)$.

Before formally stating the proof of the theorem, it will be useful to
present a natural method for `computing' $\uncertainty_i(\cT)$. This
method works bottom up, assigning a value $\beta_{v,i}$ to the node
$v$ of the tree. It will turn out that
$\beta_{\mathsf{root},i}=\uncertainty_i(\cT)$. The intermediate values
$\beta_{v,i}$ produced in this algorithm, will be used crucially when
our final sampling algorithms simulates the computation of $\cT$.  In
fact, $\beta_{v,i}$ has the following natural interpretation.
Consider a tree $\cT'$ as in the definition of $\uncertainty_i(\cT)$.
Let $D_0$ and $D_1$ be the distributions on the leaves of $\cT'$ when
$x_0$ is set to $0$ and $1$. The minimum overlap (over all possible
such $\cT'$) between $D_0$ and $D_1$ when restricted to the leaves in
the subtree rooted at $v$, is the quantity $\beta_{v,i}$. With this
interpretation, consider the following computation.
\begin{itemize} 
\item for a leaf $v$, let 
       $\beta_{v,i} =\min\{p_{v,i}(0), p_{v,i}(1)\}$.  
\item for an internal node $v$ with
children $v_L$ and $v_R$, where $x_j$ is queried, 
\[ \beta_{v,j} =
\left\{ 
\begin{array}{l l} 
\min\{\beta_{v_L,i},\beta_{v_R,i}\} & \mbox{\ \ if } i \neq j\\ 
\beta_{v_L,i}+\beta_{v_R,i} & \mbox{\ \ if } i=j 
\end{array} \right. 
\] 
\end{itemize}
The following claim, which we state without a formal proof, 
is now immediate.
\begin{proposition} For $i=1,2,\ldots,k$ we have
\begin{eqnarray*}
\beta_{\mathsf{root},i}&=&\uncertainty_{i}(\cT)\\
\uncertainty_{v,i} &\leq& \min\{p_{v,i}(0), p_{v,i}(1)\}.
\end{eqnarray*}
\end{proposition}

\newcommand{\tell}{\tilde{\ell}}

\begin{proof}
Our goal is to simulate the computation of this tree using a
randomized sampling algorithm, where input $x_i$ is read independently
with probability $1-\beta_i(\cT)$. We want to ensure that for every
input the leaf reached in the end of this simulation has the same
distribution as in the original tree $\cT$. To specify how this
simulation is to be performed, we need to determine the following.
\begin{quote}
For each internal node of $\cT$, we need the transition probabilities
for moving to each child when the input for that node is available in
the sample, and when it is not. For this, we will specify for each
internal node $v$, a function $\tell_v: \{0,1\} \rightarrow [0,1]$ and
a value $\alpha_v \in [0,1]$, which are to be used as follows. When
the computation reaches node $v$, where $x_i$ is to be read and the
value of $x_i$ is available, then the next node is $v_L$ with
probability $\tell_v(x_i)$ and $v_R$ with probability
$1-\tell_v(x_i)$; if $x_i$ is not available then the next node is
$v_L$ with probability $\alpha_v$, and $v_R$ with probability
$1-\alpha_v$.
\end{quote}
\newcommand{\tp}{\tilde{p}}
\newcommand{\tbeta}{\tilde{\beta}}
\newcommand{\event}{\mathcal{E}}

Once $\tell_v$ and $\alpha_v$ have been specified, we may consider the
events $\chi_{v}$ and $\chi_{v,i}$ as before.  Note that the events in
$\{\chi_{v,i}: i=1,2,\ldots,k\}$ are independent and their
intersection is precisely $\chi_{v}$. Let $\tp_{v}(x)$ be the
probability of $\chi_v$ in the simulation for input $x$, and similarly
let $\tp_{v,i}(z)$ be the probability of the event $\chi_{v,i}$ when
$x_i=z$ (note that the probability of $\chi_{v,i}$ depends only on
$x_i$). Clearly, $\tp_v(x)=\prod_{i=1}^k \tp_{v,i}(x_i)$, and to show
that our simulation is faithful to the original computation, it will
suffice to verify that $p_{v,i}=\tp_{v,i}$ for all $i$. The rest of
the proof consists of two steps.
\begin{description}
\item[Step 1:] Using the values $\beta_{v,i}$ defined above,
define $\tell_v$ and $\alpha_v$.

\item[Step 2:] Show that for $i\in \{1,2,\ldots,k\}$ and each node
$v \in \cT$ and $z \in \{0,1\}$, we have $\tp_{v,i}(z)=p_{v,i}(z)$.
\end{description}
We now implement this two-step plan.  Consider the first step.  Recall
the values of $\beta_{v,i}$ defined above using a bottom up
computation on the tree $\cT$.  We can now define $\alpha_v$ right
away based on the $\beta_{v,i}$'s computed above. If $v$ has label
$i$, then \[ \alpha_v= \frac{\beta_{v_L,i}}{\beta_{v,i}}.\]
Now, consider an internal node $v$.
Let $v_1,v_2,\ldots, v_r=v$ be a path in the tree from the root $v_1$
to the node $v$.  Let $\tbeta_{v,i}$ denote the probability that the
computation reaches node $v$ and $x_i$ is not sampled, assuming that
the choices at nodes not labelled $i$ do not cause the computation to
leave the path.  Thus,
\[ \tbeta_{v,i} = \Pr[\chi_{v,i} \wedge \neg\event_i]=
\beta_{\mathsf{root},i}\prod_{j=1}^{r-1} \gamma_j,\]
where  $\event_i$ is the event ``$x_i$ is sampled,'' and
\[ \gamma_j = \left\{ \begin{array}{l l}
              1 & \mbox{if $x_i$ is not queried at $v_j$}\\
           \alpha_{v_j} & \mbox{if $v_{j+1}$ is the left child of $v_j$}\\
           1-\alpha_{v_{j+1}} &
               \mbox{if $v_{j+1}$ is the right child of $v_j$}
          \end{array}
         \right. .
\]
We will show that the following choice for $\tell_v$ ensures that 
the probability of reaching every node is preserved in our simulation:
\[ \tell_v(z)= \frac{p_{v_L,i}(z) - \tbeta_{v_L,i}}
                    {p_{v,i}(z) - \tbeta_{v,i}}.\]
This completes Step 1.

Now, we move to Step 2 and verify that these definitions ensure that
$p_{v,i}(z)=\tp_{v,i}(z)$. Clearly, the claim is true for the root,
for both quantities are $1$.  Suppose the claim is true for a node
$v$. We will now show that it is true of $v_L$ and $v_R$ as well.   Consider, $v_L$. We
have
\[ p_{v_L,i}(z) = \Pr[\chi_{v_L,i} \wedge \event_i] +
                  \Pr[\chi_{v_L,i} \wedge \neg\event_i],\]
where the probabilities are computed assuming that $x_i=z$.
Using our assumption that the claim holds for $v$, we can
compute the first term as 
\[
\Pr[\chi_{v,i} \wedge \event_i] \cdot \tell_{v}(z) = 
(p_{v,i}(z)-\tbeta_{v,i}) \cdot \frac{p_{v_L,i}(z)-\tbeta_{v_L,i}}
                                    {p_{v,i}(z)-\tbeta_{v,i}}
                           = p_{v_L,i}(z)-\tbeta_{v_L,i}.\]
By definition, the second term is precisely $\tbeta_{v_L,i}$.
It follows that $p_{v_L,i}(z)=\tp_{v_L,i}(z)$, and the claim holds for 
$v_L$. A similar calculation shows that the claim holds for $v_R$ as well.
This completes Step 2.

Thus, the simulation induces the same distribution on the leaves of
$\cT$ as the original computation, and therefore computes $f$ with the
same probability of error.
\end{proof}

The following proposition states that a sampling algorithm that leaves many variables
unread cannot compute majority reliably.

\begin{proposition}
\label{prop:majority-by-sampling}
Suppose $N=2k+1$ is odd. Let $\cA$ be a sampling algorithm with inputs from
$\{0,1\}^N$, which leaves each variable unread with probability 
$N^{-\frac{1}{5}}$. Let $X$ be uniformly distributed on all strings 
with $k$ or $k+1$ ones. Then, 
$\Pr[\cA(X) = \majority(X)] \leq \frac{1}{2} + o(1)$.
\end{proposition}

\begin{proof} (Sketch.)
By the Chernoff bound, with probability $1-o(1)$, the size of the
sample picked by $\cA$ is $N-\theta(N^{4/5})$. We may assume that
$\cA$ bases its decision only on the number of $1$'s in the sample.
Let $D_0$ be the distribution of the number of $1$'s in the 
sample conditioned on the number of $1$'s in $X$ being $k$ and 
let $D_1$ be the corresponding distribution conditioned on the number
of $1$'s in $X$ being $k+1$. A direct computation shows that
the relative entropy
\[ S(D_0 \| D_1) = \sum_{i} D_0(i) \log \frac{D_0(i)}{D_1(i)}= o(1).\]
It follows that $\ell_1$ distance between $D_0$ and $D_1$ is $o(1)$.
Our claim follows from this.
\end{proof}

\begin{proof}[Proof of main theorem]
Proposition~\ref{prop:protocol-to-tree} guarantees that if we have a protocol for computing majority with constant error
that uses less than $\frac{1}{C} N \log \log N$ transmissions for some constant $C$, then we have a randomized 
decision tree to compute majority of $\sqrt{N}$ bits with constant error and  $\beta(\cT) \geq N^{-\frac{1}{10}}$. 
Theorem~\ref{thm:randomized-to-sampling} then guarantees we have a sampling algorithm $\cA$ that samples every variable
with probability $1 - N^{-\frac{1}{10}}$, and yet manages to compute the majority of $\sqrt{N}$ bits with constant error.
But this is impossible by Proposition~\ref{prop:majority-by-sampling}.
\end{proof}
\section{Conclusions}
\label{sec:conclusions}

In this paper, we presented a technique of converting computation on randomized decision tree model to 
computation on a model of sampling algorithms. We related the uncertainty of an input variable in the randomized decision tree
model to the probability that the variable is left unread by the sampling algorithm. 

We showed the power of this technique for proving lower bounds by providing elementary arguments 
to prove all the lower bounds on average noisy decision tree complexity for computing various functions
presented by Evans and Pippenger~\cite{Evans99}.

Using our technique, we then presented lower bounds for wireless
communication networks where there is a restriction on transmission power. Any bit sent by a transmitter
is received (with channel noise) only by receivers which are within the transmission radius of the transmitter.
We showed that to compute the parity and majority function of $N$ input bits with constant probability of error, we
need $\Omega(N \log \log N)$ transmissions. This result simplifies and extends the same earlier lower bound for parity~\cite{Dutta08, Dutta08-journal} 
and nicely complements the upper bound result of Ying, Srikant and Dullerud~\cite{Ying06}, which showed that $O(N \log \log N)$ transmissions are sufficient for computing the sum of all the $N$ bits. Our result also implies that the sum of $N$ bits cannot be approximated 
up to a constant additive error by any constant error protocol for $\cN(N,R)$ using $o(N \log \log N)$ transmissions, 
if $R \leq N^{-\beta}$ for some $\beta > 0$.

\bibliographystyle{alpha}
\bibliography{sampling_full}

\end{document}